\documentclass[onecolumn,draftclsnofoot,12pt]{IEEEtran}
\usepackage{amsfonts}
\usepackage{amsmath}
\usepackage{amsthm}
\usepackage{verbatim}
\usepackage{setspace}
\usepackage[caption=false,font=footnotesize]{subfig}

\newtheorem{theorem}{Theorem}

\newtheorem{proposition}{Proposition}
\newtheorem{corollary}{Corollary}

\newtheorem{remark}{Remark}

\usepackage{cite}

\ifCLASSINFOpdf
   \usepackage[pdftex]{graphicx}
\else
\fi

\def\bb0{{\mathbb{0}}}

\def\bb{{\mathbf{b}}}

\def\b0{{\mathbf{0}}}

\def\sf0{{\mathsf{0}}}

\usepackage{algorithm}
\usepackage{algpseudocode}
\usepackage{epstopdf}
\usepackage{balance}
\usepackage{bbm}
\usepackage{todo}
\IEEEoverridecommandlockouts

\begin{document}
%
\title{On Wirelessly Powered Communications with Short Packets}


\author{Talha Ahmed Khan${}^{*}$, Robert W. Heath Jr.${}^{*}$ and Petar Popovski${}^{\dagger}$\\

\thanks{${}^{*}$ Talha Ahmed Khan and Robert W. Heath Jr. are with the Department of Electrical and Computer Engineering at The University of Texas at Austin, USA (Email: \{talhakhan, rheath\}@utexas.edu).}
\thanks{${}^{\dagger}$ Petar Popovski is with the Department of Electronic Systems at Aalborg University, Denmark (Email: petarp@es.aau.dk).}
\thanks{This work was supported in part by the Army Research Office under grant W911NF-14-1-0460, and a gift from Mitsubishi Electric Research Labs.}}
\maketitle
\begin{abstract}
Wireless-powered communications will entail short packets due to naturally small payloads, low-latency requirements and/or insufficient energy resources to support longer transmissions. In this paper, a wireless-powered communication system is investigated where an energy harvesting transmitter, charged by a power beacon via wireless energy transfer, attempts to communicate with a receiver over a noisy channel. Leveraging the framework of finite-length information theory, the system performance is analyzed using metrics such as the energy supply probability at the transmitter, and the achievable rate at the receiver. The analysis yields useful insights into the system behavior in terms of key parameters such as the harvest blocklength, the transmit blocklength, the average harvested power and the transmit power. Closed-form expressions are derived for the asymptotically optimal transmit power. Numerical results suggest that power control is essential for improving the achievable rate of the system in the finite blocklength regime. 
\end{abstract}

\begin{IEEEkeywords}
Energy harvesting, wireless information and power transfer, energy supply probability, wireless power transfer, power control, finite-length information theory, non-asymptotic achievable rate.
\end{IEEEkeywords}
%
\IEEEpeerreviewmaketitle

\section{Introduction}
With wireless devices getting smaller and more energy-efficient, energy harvesting is emerging as a potential technology for powering such low-power devices.  
A related area is RF (radio frequency) or wireless energy harvesting where an energy harvesting device extracts energy from the incident RF signals \cite{EnergyHarvestWirelessCommSurvey2015}. This is attractive for future paradigms such as the Internet of Things, where powering a potentially massive number of devices will be a major challenge \cite{IoT2014}. Many of these energy harvesting devices will need to sense and communicate information to a cloud or control unit. In contrast to most wireless systems designed for Internet access, energy harvesting communication systems will likely make exclusive use of short packets. This is due to intrinsically small data payloads, low-latency requirements, and/or lack of energy resources to support longer transmissions\cite{EnergyHarvestWirelessCommSurvey2015,durisi2015towards,yang2014finite,fong2015non}.  

For an energy harvesting system with short packets, the capacity analysis conducted in the asymptotic blocklength regime could be misleading. This has spurred research characterizing the performance of an energy harvesting communication system in the non-asymptotic or finite blocklength regime \cite{fong2015non,yang2014finite,guo2016finite,ebrahim2015lowlatency}. Leveraging the finite-length information theoretic framework proposed in \cite{polyanskiy2010finite}, \cite{yang2014finite} characterized the achievable rate for a noiseless binary communications channel with an energy harvesting transmitter. This work was extended to the case of an additive white Gaussian noise (AWGN) channel and more general discrete memoryless channels in \cite{fong2015non}. For an energy harvesting transmitter operating under a save-then-transmit protocol \cite{ozel2012awgn}, the achievable rate at the receiver was characterized in the finite blocklength regime\cite{fong2015non}. 
The authors in \cite{guo2016finite} investigated the mean delay of an energy harvesting channel in the finite blocklength regime.
Unlike the work in \cite{yang2014finite,fong2015non,guo2016finite} which assume an infinite battery at the energy harvester, \cite{ebrahim2015lowlatency} conducted a finite-blocklength analysis for the case of a battery-less energy harvesting channel. 

In this paper, we analyze the performance of a wireless-powered communication system where an RF energy harvesting node, charged by a power beacon via wireless energy transfer, attempts to communicate with a receiver over an AWGN channel. Using the framework of finite-length information theory\cite{polyanskiy2010finite}, we characterize the energy supply probability and the achievable rate of the considered system for the case of short packets. Leveraging the analytical results, we expose the interplay between key system parameters such as the harvest and transmit blocklengths, the average harvested power, and the transmit power. We also provide closed-form analytical expressions for the asymptotically optimal transmit power. Numerical results reveal that the asymptotically optimal transmit power is also approximately optimal for the finite blocklength regime. 
The prior work \cite{fong2015non,yang2014finite,ebrahim2015lowlatency,guo2016finite} treating short packets falls short of characterizing the performance for the case of wireless energy harvesting. Moreover, most prior work \cite{fong2015non,yang2014finite,ebrahim2015lowlatency,guo2016finite,ozel2012awgn} implicitly assumes concurrent harvest and transmit operation, which may be infeasible in practice. 

\section{System Model}\label{secSys}
We consider a wireless-powered communication system where a wireless power beacon (PB) charges an energy harvesting (EH) device, which then attempts to communicate with another receiver (RX) using the harvested energy. 
The nodes are equipped with a single antenna each. We assume that the energy harvester uses a \emph{save-then-transmit} protocol \cite{ozel2012awgn} to enable wireless-powered communications. 
The considered protocol divides the communication frame consisting of $s$ channel uses (or slots) into an energy harvesting phase having $m$ channel uses, and an information transmission phase having $n$ channel uses. 
The first $m$ channel uses are used for harvesting energy from the RF signals transmitted by the power beacon, which is then saved in a (sufficiently large) energy buffer.
This is followed by an information transmission phase consisting of $n$ channel uses, where the transmitter uses the harvested energy to transmit information to the receiver. We call $m$ the \textit{harvest} blocklength, $n$ the \textit{transmit} blocklength, and $s=m+n$ the \textit{total} blocklength or frame size. 
 We will conduct the subsequent analysis for the non-asymptotic blocklength regime, i.e., for the practical case of \emph{short packets} where the total blocklength is finite.

\subsection{Energy Harvesting Phase} 
The signal transmitted by a power beacon experiences distance-dependent path loss and channel fading before reaching the energy harvesting node. The harvested energy is, therefore, a random quantity due to the underlying randomness of the wireless link. 
We let random variable $Z_i=\eta\beta P_{\rm{PB}} H_i$ model the energy harvested in slot $i$ ($i=1,\cdots,m$), where $\eta\in(0,1)$ denotes the harvester efficiency, $P_{\rm{PB}}$ is the PB transmit power, $\beta$ gives the average large-scale channel gain, while the random variable $H_i$ denotes the small-scale channel gain. Note that we have ignored the energy due to noise since it is negligibly small. 
We consider quasi-static block flat Rayleigh fading for the PB-EH links such that the channel remains constant over a block, and randomly changes to a new value for the next block. 
In other words, the energy arrivals within a harvesting phase are fully correlated, i.e., $Z_i=Z_1\equiv Z,\,\forall\,i=1,2,\cdots,m$, where $Z_i$ is exponentially distributed with mean $\mathbb{E}[Z_i]\triangleq P_{\rm{E}}=\eta\beta P_{\rm{PB}}$. 
This is motivated by the observation that the harvest blocklength in a \emph{short-packet} communications system would typically be smaller than the channel coherence time. 

\subsection{Information Transmission Phase} The energy harvesting phase is followed by an information transmission phase where the EH node attempts to communicate with a destination RX node over an unreliable AWGN channel. 
We assume that the EH node uses a Gaussian codebook for signal transmission (see Section \ref{secIT}).
We let $X_\ell$ be the signal intended for transmission in slot $\ell$ with average power $P_{\rm{t}}$, where $\ell=1,\cdots,n$, and $n$ is fixed. 
The resulting (intended) sequence $X^n=\left(X_1,\cdots,X_n\right)$ consists of independent and identically distributed (IID) Gaussian random variables such that $X_\ell$ {\raise.17ex\hbox{$\scriptstyle{\sim}$}}~$\mathcal{N}(0,P_{\rm{t}})$. 
To transmit the intended sequence $X^n$ over the transmit blocklength, the EH node needs to satisfy the following energy constraints. 
\begin{align}\label{eq:constraint}
	\sum_{\ell=1}^{k}X_\ell^2&\leq\sum_{i=1}^{m}Z_i\qquad k=1,2,\cdots,n
\end{align}
The constraints in (\ref{eq:constraint}) simplify to $\sum_{\ell=1}^{n}X_\ell^2\leq mZ$ for the case of correlated energy arrivals. We let $\tilde{X}^n=\left(\tilde{X}_1,\cdots,\tilde{X}_n\right)$ be the transmitted sequence. Note that $\tilde{X}^n\neq{X}^n$ when the energy constraints are violated as the EH node lacks sufficient energy to put the intended symbols on the channel.
The signal received at the destination node in slot $\ell$ is given by
$Y_\ell=\tilde{X_\ell} + V_\ell$, where $V^n=\left(V_1,\cdots,V_n\right)$ is an IID sequence modeling the receiver noise such that $V_\ell~{\raise.17ex\hbox{$\scriptstyle{\sim}$}}~\mathcal{N}(0,\sigma^2)$ is a zero-mean Gaussian random variable with variance $\sigma^2$. Note that any deterministic channel attenuation for the EH-RX link can be equivalently tackled by scaling the noise variance. Similarly, we define $Y^n=\left(Y_1,\cdots,Y_n\right)$ as the received sequence.

\subsection{Information Theoretic Preliminaries}\label{secIT}
We now describe the information theoretic preliminaries for the EH-RX link. Let us assume that the EH node transmits a message $W\in\mathcal{W}$ 
over $n$ channel uses. Assuming $W$ is drawn uniformly from $\mathcal{W}\triangleq\{1,2,\cdots,M\}$, we define an $(n,M)$-code having the following features:
It uses a set of encoding functions $\{\mathcal{F}_\ell\}_{\ell=1}^{n}$ for encoding the source message $W\in\mathcal{W}$ given the energy harvesting constraints, i.e., the source node uses $\mathcal{F}_\ell : \mathcal{W}\times\mathbb{R}^\ell_+ \rightarrow \mathbb{R}$ for transmission slot $\ell$, where
$\mathcal{F}_\ell(W,Z^\ell)=\tilde{X}_\ell$ given $Z^\ell=(Z_1,\cdots,Z_\ell)$ such that the energy harvesting constraint in (\ref{eq:constraint}) is satisfied. Specifically, $\tilde{X}_\ell={X}_\ell$ where ${X}_\ell~${\raise.17ex\hbox{$\scriptstyle{\sim}$}}~$\mathcal{N}(0,P_{\rm{t}})$ is drawn IID from a Gaussian codebook when (\ref{eq:constraint}) is satisfied, and $\tilde{X}_\ell=0$ otherwise.
It uses a decoding function $\mathcal{G}:\mathbb{R}^n\rightarrow\mathcal{W}$ that produces the output $\mathcal{G}(Y^n)=\hat{{W}}$, where $Y^n=\left(Y_1,\cdots,Y_n\right)$ is the sequence received at the destination node.

We let $\epsilon\in [0,1)$ denote the target error probability for the noisy communication link.
For $\epsilon\in[0,1)$, an $(n,M,\epsilon)$-code for an AWGN EH channel is defined as the $(n,M)$-code for an AWGN channel such that the average probability of decoding error $\Pr\{\hat{W}\neq W\}$ does not exceed $\epsilon$.
A rate $R$ is \emph{$\epsilon$-achievable} for an AWGN EH channel if there exists a sequence of {$(n,M_n,\epsilon_n)$-codes} such that $\liminf\limits_{n\rightarrow\infty}\frac{1}{n}\log(M_n) \geq R$ and $\limsup\limits_{n\rightarrow\infty} \epsilon_n \leq \epsilon$.
The \emph{$\epsilon$-capacity} $C_\epsilon$ for an AWGN EH channel is defined as $C_\epsilon=\sup\{R:R\,\,\text{is}\, \epsilon\text{-\textit{achievable}}\}$. 	
\subsection{Performance Metrics} 
We now introduce the metrics used for characterizing the performance of the considered \emph{short-packet} wireless-powered communications system. 
Note that the overall performance is marred by two key events. First, due to lack of sufficient energy, the EH node may not be able to transmit the intended codewords during the information transmission phase, possibly causing a decoding error at the receiver. Second, due to  a noisy EH-RX channel, the received signal may not be correctly decoded. For the former, we define a metric called the \emph{energy supply probability}, namely, the probability $\Pr\left[\sum_{i=1}^{n}X_i^2\leq mZ\right]$ that an EH node can support the intended transmission. 
For the latter, we define and characterize the \emph{$\epsilon$-achievable rate} in the finite blocklength regime.

\section{Analytical Results}\label{secAnl}
In this section, we characterize the energy supply probability and the achievable rate in the finite blocklength regime.
\subsection{Energy Supply Probability}
We define the \textit{energy supply probability} $P_\textrm{es}(m,n,a)$ as the probability that an EH node has sufficient energy to transmit the intended codeword, namely, 
\begin{align}\label{eq:def}
P_\textrm{es}(m,n,a)=\Pr\left[\sum_{i=1}^{n}X_i^2\leq mZ\right]
\end{align}
for a harvest blocklength $m$, a transmit blocklength $n$, and a power ratio $a=\frac{P_{\rm{t}}}{P_{\rm{E}}}$. Similarly, we define $P_\textrm{eo}(m,n,a)=1-P_\textrm{es}(m,n,a)$  as the \textit{energy outage probability} at the energy harvesting node. The following proposition characterizes the energy supply probability for the considered system.

\begin{proposition}\normalfont
Assuming the intended transmit symbols $\{X_i\}_{i=1}^{n}$ are drawn IID from $\mathcal{N}(0,P_t)$, the energy sequence $\{Z_i\}_{i=1}^{m}=Z$ is fully correlated, and $Z$ follows an exponential law with mean $P_{\rm{E}}$, the energy supply probability is given by
\begin{align}\label{eq:p_ec single PB}
	P_\textrm{es}(m,n,a)=\frac{1}{\left(1+\frac{2a}{m}\right)^{\frac{n}{2}}}
\end{align}
for $m>2a$ where $a=\frac{P_\textrm{t}}{P_{\rm{E}}}$, while $m$ and $n$ denote the blocklengths for the harvest and the transmit phase.
\end{proposition}
\begin{proof}
	The proof follows by leveraging the statistical properties of the random variables. Consider
	\begin{align}\label{eq:prop 1}
	P_\textrm{es}\left(m,n,a\right)&=\Pr\left[\sum_{i=1}^{n}X_i^2\leq mZ\right]
	\overset{(a)}{=}\Pr\left[W \leq \frac{mZ}{P_\textrm{t}}\right]\nonumber\\
	&\overset{(b)}{=}\mathbb{E}_{W}\left[e^{-\frac{P_\textrm{t}}{P_{\rm{E}}m}W}\right]
	\overset{}{=}\frac{1}{\left(1+\frac{2a}{m}\right)^\frac{n}{2}}
	\end{align}
	where (${a}$) follows from the substitution $W=\frac{1}{P_\textrm{t}}\sum_{i=1}^{n}X_i^2$ where $W$ is a Chi-squared random variable with $n$ degrees of freedom. (${b}$) is obtained by conditioning on the random variable $W$, and by further noting that $Z$ is exponentially distributed with mean $P_{\rm{E}}$. Assuming $m>2a$, the last equation follows from the definition of the moment generating function of a Chi-squared random variable.
\end{proof}
While Proposition 1 is valid for $m>2a$, we note that this is the case of practical interest since it is desirable to operate at $a<1$, as evident from Section \ref{secSim}. Further, the expression in (\ref{eq:p_ec single PB}) makes intuitive sense as the energy outages would increase with the transmit blocklength $n$ for a given $m$, and decrease with the harvest blocklength $m$ for a given $n$. 
Let us fix $P_{\rm{t}}$ and $P_{\rm{E}}$. For a given $m$, we may improve the reliability of the EH-RX communication link by increasing the blocklength $n$, albeit at the expense of the energy supply probability. 
With a smaller transmit power $P_{\rm{t}}$, the energy harvester is less likely to run out of energy during an ongoing transmission. Therefore, when $m+n$ is fixed, we may reduce $P_{\rm{t}}$ to meet the energy supply constraint, but this would reduce the channel signal-to-noise ratio (SNR). 
This underlying tension between the energy availability and the communication reliability will be highlighted throughout the rest of this paper.
The following discussion relates the transmit power to the harvest and transmit blocklengths, illustrating some of the key tradeoffs.  
\begin{remark}\normalfont
The energy supply probability is more sensitive to the length of the transmit phase compared to that of the harvest phase. 
This observation also manifests itself in terms of the energy requirements at the transmitter. For instance, to maintain an energy supply probability $\rho$, it follows from (\ref{eq:p_ec single PB}) that the power ratio satisfies $a\geq\frac{m}{2}\left(\rho^{-\frac{2}{n}}-1\right)$. Note that the power ratio varies only linearly with the harvest blocklength $m$, but superlinearly with the transmit blocklength $n$. This further implies that for a fixed $n$, doubling the harvest blocklength relaxes the transmit power budget by the same amount. That is, the energy harvester can double its transmit power $P_{\rm{t}}$ (and therefore the channel SNR) without violating the required energy constraints. In contrast, reducing the transmit blocklength for a given $m$ brings about an exponential increase in the transmit power budget at the energy harvester. 
\end{remark}
The following corollary treats the scaling behavior of the energy supply probability as the  blocklength becomes large.
\begin{corollary}\normalfont
When the harvest blocklength $m$ scales in proportion to the transmit blocklength $n$ such that $m=cn$ for some constant $c>0$, the energy supply probability $P_\textrm{es}(m,n,a)$ converges to a limit as $n$ becomes asymptotically large. In other words, $\lim\limits_{n\rightarrow\infty}P_\textrm{es}(m,n,a)=e^{-\frac{a}{c}}<1$ such that the limit only depends on the power ratio $a>0$ and the proportionality constant $c>0$. Further, under proportional blocklength scaling, this limit also serves as an upper bound on the energy supply probability for finite blocklengths, i.e., $P_\textrm{es}(m,n,a)\leq e^{-\frac{a}{c}}<1$.
\end{corollary}
The previous corollary also shows that energy outage is a fundamental bottleneck regardless of the blocklength, assuming at best linear scaling.

\subsection{Achievable Rate}
The following result characterizes the $\epsilon$-achievable rate of the considered wireless-powered communication system in the finite blocklength regime.
\begin{theorem}\normalfont
Given a target error probability $\epsilon\in[0,1)$ for the noisy channel, the $\epsilon$-achievable rate $R_{EH}^{}\left(\epsilon,m,n,a,\gamma\right)$ of the considered system with harvest blocklength $m$, transmit blocklength $n$, power ratio $a$ (where $2a<m$), and the SNR $\gamma=\frac{P_{\rm{t}}}{\sigma^2}$ is given by 
\begin{align}\label{eq:main}
	R_{EH}^{}\left(\epsilon,m,n,a,\gamma\right)=\frac{\frac{n\log(1+\gamma)}{2}
		-\sqrt{\frac{2+\epsilon}{\epsilon}\frac{\gamma}{\gamma+1}n}
			-{(n)}^{\frac{1}{4}}-1}
			{n+m}
\end{align}
for all tuples $(m,n)$ satisfying
\begin{align}\label{eq:mainc1}
m\geq\frac{2a}{\exp\left(\frac{2\ln(1+0.5\epsilon)}{\left(\ln\left[\frac{2+\epsilon}
		{\epsilon^2}\right]\right)^4}\right)-1}
\end{align}
and
\begin{align}\label{eq:mainc2}
n\leq 2 \frac{\ln(1+0.5\epsilon)}{\ln\left(1+\frac{2a}{m}\right)}.
\end{align}	
\end{theorem}
\begin{proof}
	See Appendix.
\end{proof}
For a given target error probability $\epsilon$, 
a harvest blocklength $m$ can support a transmit blocklength only as large as in (\ref{eq:mainc2}). Moreover, a sufficiently large $m$, as given in (\ref{eq:mainc1}), is required for a sufficiently large $n$ to meet the target error probability $\epsilon$. 
The following proposition provides an analytical expression for the achievable rate in the asymptotic blocklength regime. We note that the asymptotic results provide a useful analytical handle for the non-asymptotic case as well. 
\begin{proposition}\normalfont\label{proposition:asym rate}
Let $R_{\textrm{EH}}^{\infty}(\epsilon,a,\gamma)$ denote the asymptotic achievable rate as the transmit blocklength $n\rightarrow\infty$, i.e., $R^{\infty}_{\textrm{EH}}(\epsilon,a,\gamma)=\lim\limits_{n\rightarrow\infty}R_\textrm{EH}(\epsilon,m,n,a,\gamma)$. It is given by
\begin{align} \label{eq:prop2}
R^\infty_{\textrm{EH}}(\epsilon,a,\gamma)&= L(a,\epsilon) C^\infty_{\textrm{AWGN}}(\gamma)
\end{align}
where \begin{equation}C^\infty_{\textrm{AWGN}}(\gamma)=\frac{1}{2}\log(1+\gamma),\quad\gamma\geq 0 
\end{equation}
denotes the capacity of an AWGN channel without the energy harvesting constraints, whereas
\begin{align}\label{eq:prpo2a}
L(a,\epsilon)=\frac{1}{1+\frac{a}{\log\left(1+0.5\epsilon\right)}},\quad a\geq 0,\,\epsilon\in[0,1)
\end{align}
where $L(a,\epsilon)\in[0,1]$ such that $1-L(a,\epsilon)$ gives the (fractional) loss in capacity due to energy harvesting constraints.
\end{proposition}
\begin{proof}
Using (\ref{eq:main}), $R_{\textrm{EH}}^{\infty}\left(\epsilon,a,\gamma\right)$ can be expressed as  
\begin{align}
	R_{\textrm{EH}}^{\infty}\left(\epsilon,a,\gamma\right)&=\lim\limits_{n\rightarrow \infty}\frac{\frac{n\log(1+\gamma)}{2}
		-\sqrt{\frac{2+\epsilon}{\epsilon}\frac{\gamma}{\gamma+1}n}
			-{(n)}^{\frac{1}{4}}-1}
			{n+m}\\
	&\overset{(a)}{=}\lim\limits_{n\rightarrow \infty}\frac{1}{1+\frac{m}{n}}\frac{\log(1+\gamma)}{2}\\
	&\overset{(b)}{=}\lim\limits_{n\rightarrow \infty}\frac{1}{1+\frac{2a}{n[1+0.5\epsilon]^{\frac{2}{n}}-1}}\frac{\log(1+\gamma)}{2}\\
	&\overset{(c)}{=}\underbrace{\frac{1}{1+\frac{a}{\log\left(1+0.5\epsilon\right)}}}_{L(a,\epsilon)}
	\underbrace{\frac{\log(1+\gamma)}{2}}_{C_{\textrm{AWGN}}^\infty(\gamma)}
\end{align}
where (${a}$) follows since the higher order terms in (\ref{eq:main}) vanish as $n\rightarrow\infty$. Note that for a given $\epsilon$ and $a$, $m$ and $n$ should satisfy (\ref{eq:mainc1}) and (\ref{eq:mainc2}). (b) is obtained by substituting $m=\frac{2a}{\left[1+0.5\epsilon\right]^{\frac{2}{n}}-1}$ from (\ref{eq:mainc2}), and by further assuming that $n\geq\left(\log\left(\frac{2+\epsilon}{\epsilon^2}\right)\right)^4$. Finally, (c) follows by noting that $\lim\limits_{n\rightarrow\infty}n\left(\left(1+x\right)^{\frac{2}{n}}-1\right)=2\log(1+x)$.
\end{proof}
\begin{remark}\normalfont
Proposition 2 reveals a fundamental communications limit of the considered wireless-powered system. In order to guarantee an $\epsilon$-reliable communication over $n$ channel uses, the node first needs to accumulate sufficient energy during the initial harvesting phase. A sufficiently large $m$ helps improve the energy availability at the transmitter. This harvesting overhead, however, causes a rate loss (versus a non-energy harvesting system) as the first $m$ channel uses are reserved for harvesting. Moreover, as the transmit blocklength $n$ grows, so does the length of the initial harvesting phase $m$, resulting in an inescapable performance limit on the communication system. This limit depends on i) the power ratio $a$, and ii) the required reliability $\epsilon$, and is captured by the prelog term $L(a,\epsilon)$ in (\ref{eq:prpo2a}) for a given $\gamma$. Moreover, this behavior is more visible for latency-constrained systems where the total blocklength is fixed. 
\end{remark}
\begin{corollary}\normalfont
As the power ratio $a\rightarrow0$ in (\ref{eq:prop2}), the asymptotic achievable rate converges to the capacity of a non-energy harvesting AWGN channel, i.e., $\lim\limits_{a\rightarrow0}R^\infty_{\textrm{EH}}(\epsilon,a,\gamma)=C^\infty_{\textrm{AWGN}}(\gamma)$. 
\end{corollary}
For optimal performance, the energy harvesting node needs to use the \emph{right} amount of transmit power. 
On the one hand, reducing $P_{\rm{t}}$ helps improve the energy supply probability as a packet transmission is less likely to face an energy outage. On the other hand, it is detrimental for the communication link as it lowers the SNR. We now quantify the optimal transmit power that maximizes the asymptotic achievable rate for a given set of parameters. We note that many of the analytical insights obtained for the asymptotic regime are also useful for the non-asymptotic regime (see Remark 3).

\begin{corollary}\normalfont
For a given $P_{\rm{E}}$, there exists an optimal transmit power that maximizes the achievable rate. We let $P_{\rm{t},\infty}^*$ be the rate-maximizing transmit power in the asymptotic blocklength regime. It follows that 
\begin{align}\label{eq: opt P_t}
P_{\rm{t},\infty}^*(\epsilon,P_{\rm{E}},\sigma^2)&=
\sigma^2\left(\frac{\frac{P_{\rm{E}}}{\sigma^2}\log(1+0.5\epsilon)-1}
{
\text{W}\left[\left(\frac{P_{\rm{E}}}{\sigma^2}\log(1+0.5\epsilon)-1\right) e^{-1}\right]
}-1\right)
\end{align}where $\text{W}[\cdot]$ is the Lambert W-function. 
\end{corollary}
Plugging $P_{\rm{t}}=P_{\rm{t},\infty}^*$ in Proposition \ref{proposition:asym rate} gives the optimal achievable rate in the asymptotic blocklength regime.
Furthermore, when $P_{\rm{t}}$ is fixed, the achievable rate improves monotonically with $P_{\rm{E}}$ due to an increase in the energy supply probability. 

\begin{remark}\normalfont
The optimal transmit power for the asymptotic case serves as a conservative estimate for the optimal transmit power for the non-asymptotic case (Fig. \ref{fig:P_opt vs Pe}). Moreover, the achievable rate in the non-asymptotic regime obtained using the asymptotically optimal transmit power, gives a tight lower bound for the optimal achievable rate in the non-asymptotic regime (Fig. \ref{fig:rate vs. error}). This suggests that Corollary 3 provides a useful analytical handle for transmit power selection even for the finite blocklength regime (despite the fact that the resulting rate for the non-asymptotic case could be much smaller than that for the asymptotic case).  
\end{remark}

\section{Numerical Results}\label{secSim}
We now present the simulation results based on the analysis in Section \ref{secAnl}.
The following plots have been obtained using the minimum latency approach where the minimum possible blocklength is selected for the given set of parameters, based on the constraints in (\ref{eq:mainc1}) and (\ref{eq:mainc2}). We assume the noise power $\sigma^2=1$, and do not specify the units of $P_{\rm{t}}$ and $P_{\rm{E}}$ since the results are valid for any choice of the units (say Joules/symbol).
In Fig. \ref{fig:rate vs ratio}, we use Theorem 1 and Proposition 2 to plot the achievable rate versus the power ratio $a$ for a given $\epsilon$ and $P_{\rm{E}}$. The plots reflect the underlying tension between the energy supply probability and the channel SNR, resulting in an optimal transmit power (or power ratio) that maximizes the achievable rate. 
In Fig. \ref{fig:rate vs. error}, we plot the achievable rate versus the target error probability $\epsilon$ for a  given power ratio $a$. 
 We first consider the (fixed power) case where we fix the transmit power $P_{\rm{t}}=1.1554$ and the power ratio $a=0.0012$ (these values are asymptotically optimal for $P_{\rm{E}}=10^3$ and $\epsilon=10^{-3}$).
 As $\epsilon$ increases, the achievable rate tends to increase until a limit, beyond which the rate tends to decrease. This is because as we allow for more error ($\epsilon \uparrow$), the required total blocklength decreases. This means a possible increase in the energy supply probability (as the power ratio is fixed), and a larger backoff from capacity due to a shorter transmit blocklength. Beyond a certain $\epsilon$, further reduction in blocklength pronounces the higher order backoff terms, eventually reducing the rate. For a \emph{fixed} total blocklength, however, the achievable rate indeed increases with $\epsilon$. Note that these trends differ from the asymptotic case where the rate monotonically increases with $\epsilon$. 
 
 We then consider the case where we adapt the transmit power using Corollary 3. In Fig. \ref{fig:rate vs. error}, we observe a substantial increase in the rate by optimally adjusting the transmit power in terms of the system parameters. Moreover, the \emph{asymptotically} optimal transmit power $P^*_{\rm{t,\infty}}$ (from Corollary 3) in the finite blocklength regime results in only a minor loss in performance. As evident from Fig. \ref{fig:rate vs. error}, the optimal rate in the finite blocklength regime (obtained by numerically optimizing over $P_{\rm{t}}$) is almost indistinguishable from the lower bound obtained using the asymptotically optimal power $P^*_{\rm{t,\infty}}$.

\begin{figure} [t]
	\centerline{
		\includegraphics[width=.9\columnwidth]{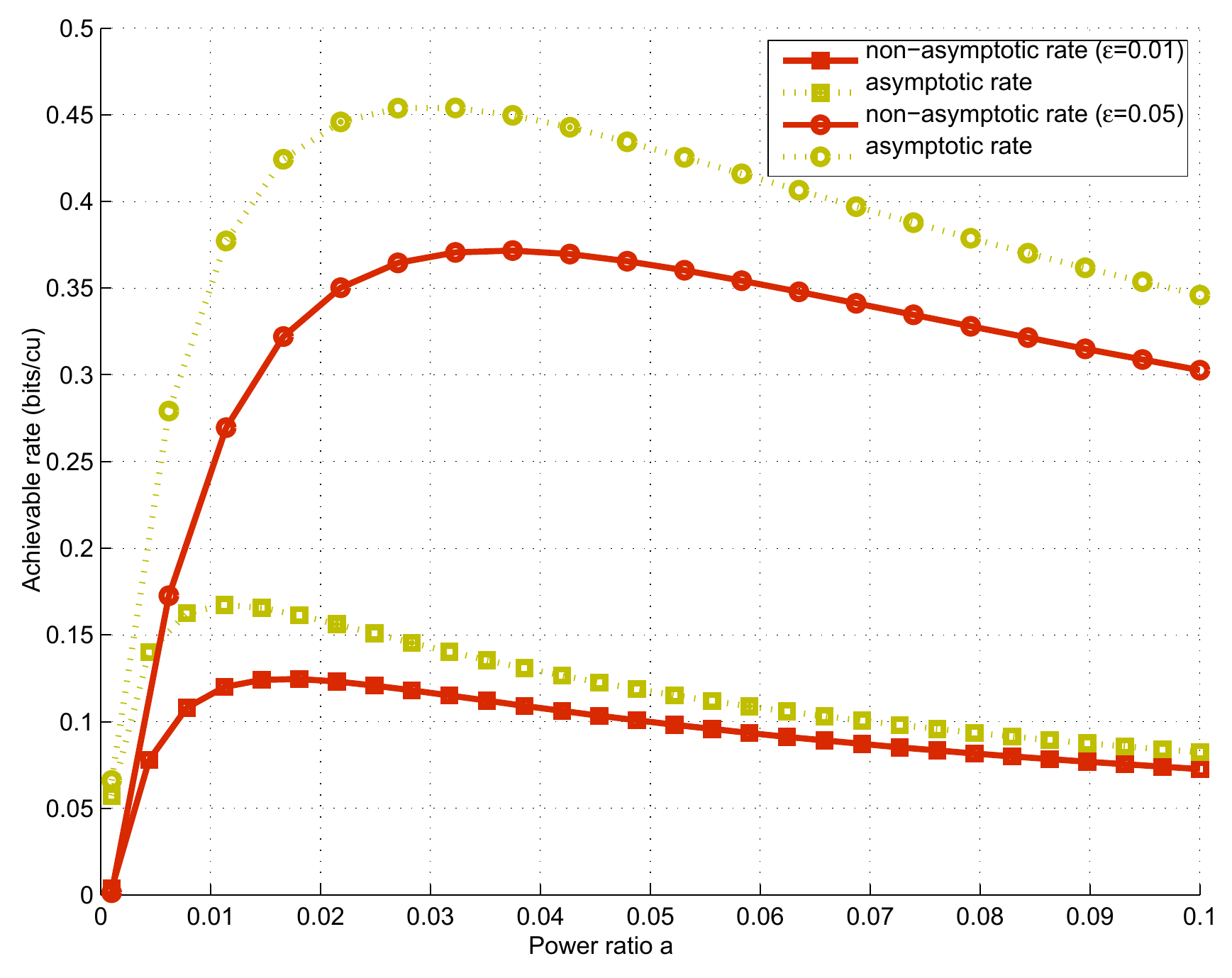}}
	\caption{The achievable rate (bits/channel use) vs. the power ratio $a=\frac{P_{\rm{t}}}{P_{\rm{E}}}$ for $P_{\rm{E}}={10}^2$. There is an optimal transmit power that maximizes the rate.}
	\label{fig:rate vs ratio}
\end{figure}

\begin{figure} [t]
	\centerline{
		\includegraphics[width=0.9\columnwidth]{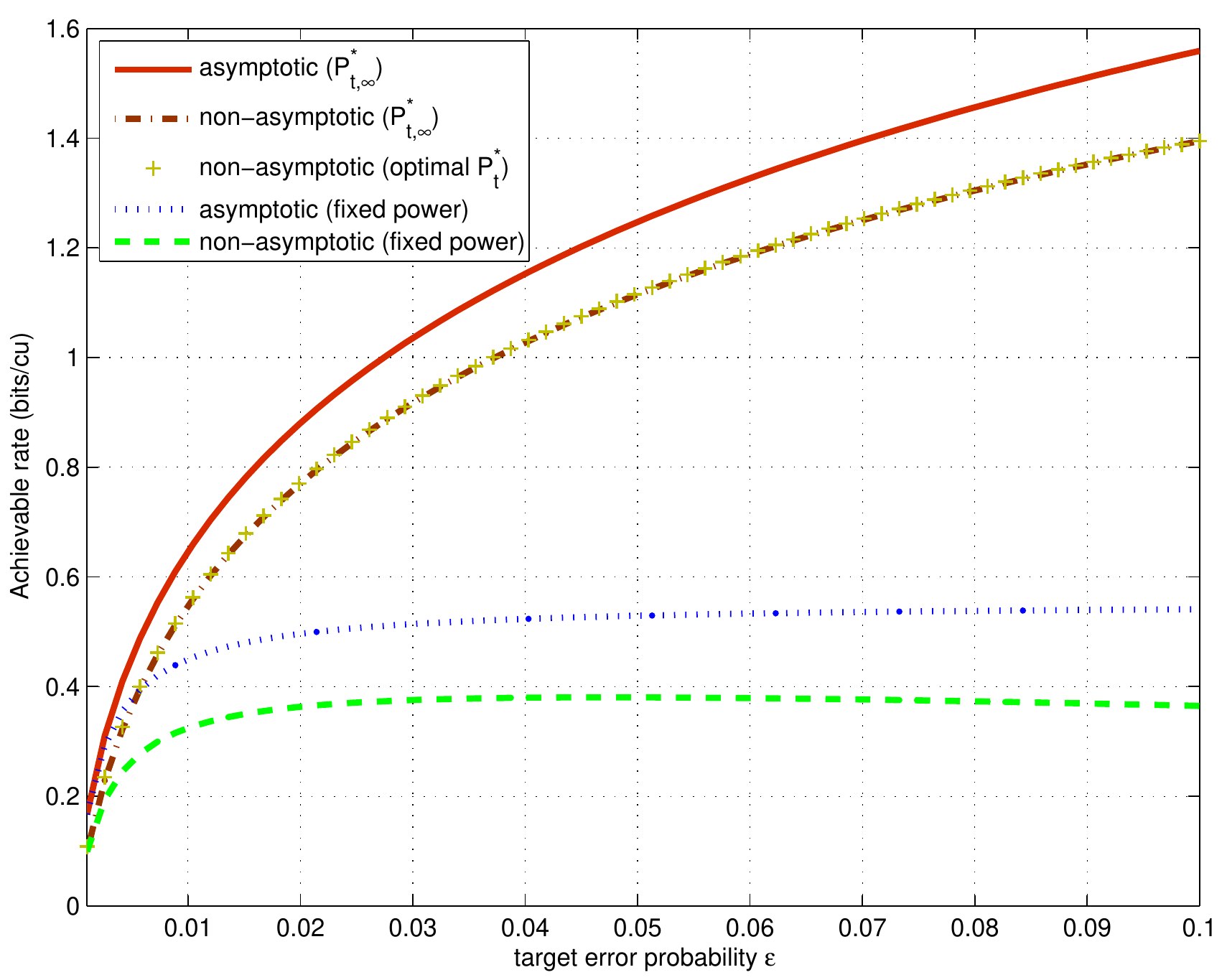}}
	\caption{The achievable rate (bits/channel use) vs. the target error probability $\epsilon$ for a given power ratio $a=0.0012$. While the asymptotic rate increases as we allow for more error, the non-asymptotic rate behaves differently.}
	\label{fig:rate vs. error}
\end{figure}

\begin{figure} [t]
	\centerline{	\includegraphics[width=.9\columnwidth]{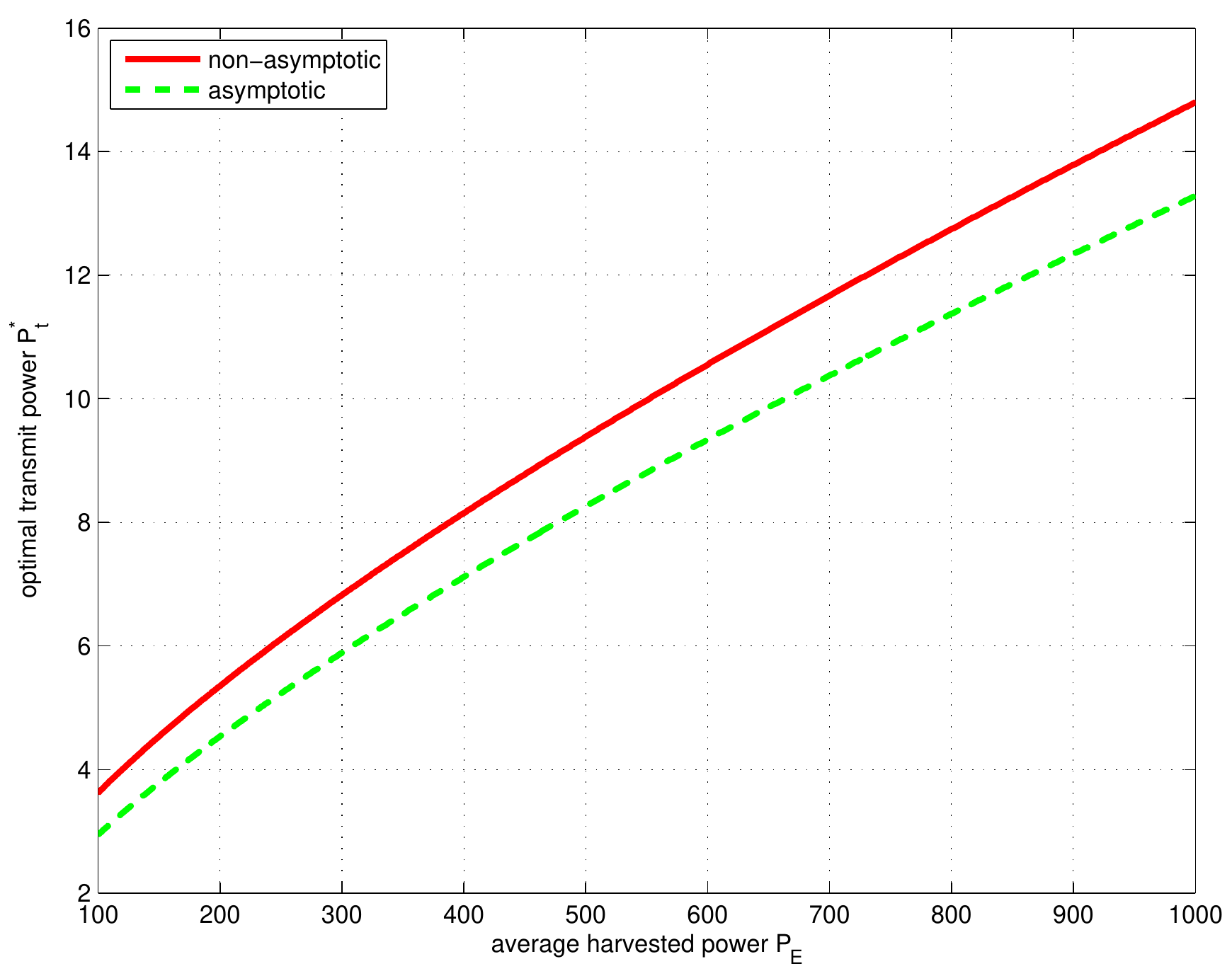}}
	\caption{Optimal transmit power $P_{\rm{t}}$ vs. average harvested power $P_{\rm{E}}$.}
	\label{fig:P_opt vs Pe}
\end{figure}

In Fig. \ref{fig:P_opt vs Pe}, we plot the optimal transmit power versus the average harvested power for $\epsilon=0.05$ and transmit blocklength $n=\lfloor \log\left(\frac{2+\epsilon}{\epsilon^2}\right)^4\rfloor=2026$. For each $P_{\rm{E}}$, the harvest blocklength is selected to satisfy the constraints in (\ref{eq:mainc1}) and (\ref{eq:mainc2}). We observe that the asymptotically optimal transmit power is a conservative estimate of the optimal transmit power for the finite case (Remark 3). Even though the optimal transmit power increases with $P_{\rm{E}}$, 
we note that the optimal power ratio still decreases as $P_{\rm{E}}$ is increased. In other words, while it is optimal to increase $P_{\rm{t}}$ with $P_{\rm{E}}$, the scaling is sublinear in $P_{\rm{E}}$. 
 
\section{Conclusions}\label{secConc}
We analytically characterized the energy supply probability and the achievable rate of a wireless-powered communication system in the finite blocklength regime. Using analytical expressions as well as numerical simulations, we investigated the interplay between key system parameters such as the harvest blocklength, the transmit blocklength, the error probability, and the power ratio. Moreover, we derived closed-form expression for the optimal transmit power in the asymptotic blocklength regime. 
Numerical results show that using the asymptotically optimal transmit power can substantially improve the achievable rate even in the finite blocklength regime.
\appendices
\allowdisplaybreaks
\section*{Appendix}

Let us bound the energy outage probability as
\begin{align}\label{eq:eneergyBound1}
\Pr\left[{\bigcup_{k=1}^{n}}\left\{\sum_{\ell=1}^{k}X_\ell^2\geq\sum_{i=1}^{m}Z_i\right\} \right]&\leq 1-\frac{2}{2+\epsilon}
\end{align}
for $\epsilon\in[0,1)$. The constraint in (\ref{eq:eneergyBound1}) can be equivalently expressed in terms of the energy supply probability as
$\Pr\left[\sum_{\ell=1}^{n}X_\ell^2\leq\sum_{i=1}^{m}Z_i \right]
\geq \frac{2}{2+\epsilon}$.
We let $X^n(W)$ and $Y^n$ denote the intended codeword sequence for a message $W\in\mathcal{W}$, and the received sequence. The decoder $\mathcal{G}({Y}^n)$ employs the following threshold decoding rule \cite{fong2015non} to decode the received signal: $\mathcal{G}({Y}^n)=i$ if there exists a unique integer $i\in\mathcal{W}$ that satisfies
\begin{align}\label{eq:rule}
\log\left(\frac{p_{Y^n|X^n}\left(Y^n|X^n(i)\right)}{p_{Y^n}\left(Y^n\right)}\right)>\log(M)+n^{\frac{1}{4}},
\end{align}
otherwise $\mathcal{G}({Y}^n)=w$, where $w$ is drawn uniformly at random from $\mathcal{W}$. 
Here, the notation $p_{Y^n|X^n}(\cdot)$ denotes the joint conditional distribution of random sequence $Y^n$ given $X^n$.
We express the probability of decoding error
$\Pr\left[\mathcal{G}({Y}^n)\neq W\right]$ in (\ref{eq:decodingerror}).
\begin{align}\label{eq:decodingerror}
\Pr&\left[\mathcal{G}({Y}^n)\neq W\right]=\nonumber\\
&\Pr\left[\mathcal{G}({Y}^n)\neq W, {Y}^n=X^n(W)+V^n\right]
+\Pr\left[\mathcal{G}({Y}^n)\neq W, {Y}^n\neq X^n(W)+V^n\right]\nonumber\\
&\leq\Pr\left[\mathcal{G}(X^n(W)+V^n)\neq W\right]+
\frac{\epsilon}{2+\epsilon},
\end{align}
where the inequality results from (\ref{eq:eneergyBound1}).
To calculate $\Pr\left[\mathcal{G}(X^n(W)+V^n)\neq W\right]$, we define ${\mathcal{A}}_{i|j}$ as the event that $i\in\mathcal{W}$ satisfies the threshold decoding rule of (\ref{eq:rule}) when $j\in\mathcal{W}$ is transmitted, i.e., 
\begin{align}\label{eq:Aij}
{\mathcal{A}}_{i|j}=\left\{\log\left(\frac{p_{Y^n|X^n}\left(X^n(j)+V^n|X^n(i)\right)}{p_{Y^n}\left(X^n(j)+V^n\right)}\right)>\log(M)+n^{\frac{1}{4}}\right\},
\end{align}
and ${\mathcal{A}}_{i|j}^c$ denotes its complement. As the message $W$ is uniform on $\mathcal{W}$, it follows that the decoding error probability 

\allowdisplaybreaks\begin{align}
\Pr\left[\mathcal{G}(X^n(W)+V^n)\neq W\right]
&\overset{(a)}{=}\frac{1}{M}\sum_{w=1}^{M}\Pr\left[{\mathcal{A}}_{w|w}^c\bigcup\bigcup_{i\neq w, i\in\mathcal{W}}{\mathcal{A}}_{i|w}\Big|W=w\right]\nonumber\\
&\overset{(b)}{=}\Pr\left[{\mathcal{A}}_{1|1}^c\bigcup\bigcup_{i=2}^M{\mathcal{A}}_{i|1}\right]
\overset{(c)}{\leq}\Pr\left[{\mathcal{A}}_{1|1}^c\right]+\Pr\left[\bigcup_{i=2}^M{\mathcal{A}}_{i|1}\right]\nonumber\\
&\overset{(d)}{\leq}\Pr\left[{\mathcal{A}}_{1|1}^c\right]+e^{-n\delta}
\overset{(e)}{\leq}\Pr\left[{\mathcal{A}}_{1|1}^c\right]+\frac{\epsilon^2}{2+\epsilon}
\end{align}
where (b) follows from the symmetry in random codebook construction, (c) results from applying the Union bound, and (d) is obtained by invoking Lemma 3 from \cite{fong2015non}. Finally, (e) follows by setting $n\delta=n^{\frac{1}{4}}$, and by further noting that 
$n\geq \left(\log\left(\frac{2+\epsilon}{\epsilon^2}\right)\right)^4$,
which follows from the constraints in (\ref{eq:mainc1}), (\ref{eq:mainc2}).
Before proceeding further, 
let us assume that $M$ is a unique integer that satisfies (\ref{eq:assumeM}).
To find a bound for $\Pr\left[{\mathcal{A}}_{1|1}^c\right]$, consider the set of inequalities in (\ref{eq:q1}), 
where ($a$) follows from the definition of ${\mathcal{A}}_{i|j}$ in (\ref{eq:Aij}), while the bound in ($b$) results from (\ref{eq:assumeM}). Finally, ($c$) is obtained by applying Chebychev's inequality. From (\ref{eq:Aij}) and (\ref{eq:q1}), it follows that 
$\Pr\left[\mathcal{G}(X^n(W)+V^n)\neq W\right]
=\frac{\epsilon+\epsilon^2}{2+\epsilon}$;
and further using (\ref{eq:decodingerror}), we conclude that 
$\Pr\left[\mathcal{G}({Y}^n)\neq W\right]\leq\epsilon$,
where $W$ is the transmitted message. Therefore, we conclude that the constructed code is an $(n+m,M,\epsilon)$-code that satisfies the following equations (\ref{eq:c1})-(\ref{eq:c3}).
\newcounter{MYtempeqncnt}
\begin{figure*}[!t]
\normalsize
\setcounter{MYtempeqncnt}{\value{equation}}
\label{eqn_dbl_x}
\setcounter{equation}{20}

\begin{align}\label{eq:assumeM}
\log(M+1)\geq n\mathbb{E}\left[\log\left(\frac{p_{Y|X}(Y|X))}{p_{Y}(Y)}\right)\right]
-\left({\frac{2+\epsilon}{\epsilon}n\mathtt{{Var}}\left[\log\left(\frac{p_{Y|X}(Y|X))}{p_{Y}(Y)}\right)\right]}\right)^{\frac{1}{2}}>\log(M)
\end{align}

\begin{align}
\label{eq:q1}
\Pr\left[{\mathcal{A}}_{1|1}^c\right]&\overset{(a)}{=}\Pr\left[\log\left(\frac{p_{Y^n|X^n}(X^n(1)+V^n|X^n(1))}{p_{Y^n}(X^n(1)+V^n)}\right)\leq\log(M)+n^{\frac{1}{4}}\right]\nonumber\\
&=\Pr\left[\sum_{k=1}^{n}\log\left(\frac{p_{Y|X}(X_k(1)+V_k|X_k(1))}{p_{Y}(X_k(1)+V_k)}\right)\leq\log(M)+n^{\frac{1}{4}}\right]\nonumber\\
&\overset{(b)}{\leq}\Pr\Bigg[\sum_{k=1}^{n}\log\left(\frac{p_{Y|X}(X_k(1)+V_k|X_k(1))}{p_{Y}(X_k(1)+V_k)}\right)\leq\nonumber\\
&\qquad\qquad\qquad n\mathbb{E}\left[\log\left(\frac{p_{Y|X}(Y|X))}{p_{Y}(Y)}\right)\right]-\left(\frac{2+\epsilon}{\epsilon}n\textrm{Var}\left[{\log\left(\frac{p_{Y|X}(Y|X))}{p_{Y}(Y)}\right)}\right]\right)^{\frac{1}{2}}\Bigg]\nonumber\\
&\leq\Pr\Bigg[\left|\sum_{k=1}^{n}\log\left(\frac{p_{Y|X}(X_k(1)+V_k|X_k(1))}{p_{Y}(X_k(1)+V_k)}\right)-
n\mathbb{E}\left[\log\left(\frac{p_{Y|X}(Y|X))}{p_{Y}(Y)}\right)\right]
\right|\geq\nonumber\\
&\qquad\qquad\qquad\qquad\qquad\qquad\left(\frac{2+\epsilon}{\epsilon}n\textrm{Var}\left[{\log\left(\frac{p_{Y|X}(Y|X))}{p_{Y}(Y)}\right)}\right]\right)^{\frac{1}{2}}\Bigg]\nonumber\\
&\overset{(c)}{\leq}\frac{\epsilon}{2+\epsilon}
\end{align}
\setcounter{equation}{\value{MYtempeqncnt}}
\hrulefill
\vspace*{4pt}
\end{figure*}

\addtocounter{equation}{2}

\begin{align}\label{eq:c1}
\log(M+1)\geq n\mathbb{E}\left[\log\left(\frac{p_{Y|X}(Y|X))}{p_{Y}(Y)}\right)\right]-\left({\frac{2+\epsilon}{\epsilon}n\mathtt{{Var}}\left[\log\left(\frac{p_{Y|X}(Y|X))}{p_{Y}(Y)}\right)\right]}\right)^{\frac{1}{2}}
\end{align}

\begin{align}\label{eq:c2}
\log(M+1)&\geq \frac{n}{2}\log(1+\gamma)
-\sqrt{\frac{2+\epsilon}{\epsilon}\frac{\gamma}{1+\gamma}n}-n^{\frac{1}{4}}\\ 
\label{eq:c3}\log(M)&\geq \frac{n}{2}\log(1+\gamma)
-\sqrt{\frac{2+\epsilon}{\epsilon}\frac{\gamma}{1+\gamma}n}-n^{\frac{1}{4}} -1 
\end{align}
Here, (\ref{eq:c2}) is obtained by noting that the mutual information $\mathbb{E}\left[\log\left(\frac{p_{Y|X}(Y|X))}{p_{Y}(Y)}\right)\right]=\frac{1}{2}\log\left(1+\gamma\right)$, while the variance $\mathtt{{Var}}\left[\log\left(\frac{p_{Y|X}(Y|X))}{p_{Y}(Y)}\right)\right]=\frac{\gamma}{1+\gamma}$. The last equation follows by noting that $\log\left(M+1\right)-\log\left(M\right)<1$. Using (\ref{eq:c3}) with the constraints in (\ref{eq:mainc1}) and (\ref{eq:mainc2}) completes the proof.

\bibliographystyle{ieeetr}

\end{document}